\documentclass{elsarticle}

\usepackage[ruled]{algorithm2e}
\usepackage{amsmath}
\usepackage{amsthm}

\SetAlFnt{\small}
\SetAlCapFnt{\small}
\SetAlCapNameFnt{\small}
\SetAlCapHSkip{0pt}
\IncMargin{-\parindent}

\usepackage{lipsum}
\makeatletter
\def\ps@pprintTitle{%
 \let\@oddhead\@empty
 \let\@evenhead\@empty
 \def\@oddfoot{}%
 \let\@evenfoot\@oddfoot}
\makeatother

\usepackage{amssymb,amsmath,txfonts}
\usepackage{enumitem}
\usepackage{mathrsfs}
\usepackage{hyphenat}
\usepackage{color}

\usepackage{comment,amssymb,amsfonts}

\newcommand{\ta}{{\bf a}}

\newcommand{\tv}{{\bf v}}

\newcommand{\hornsat}[1]{{#1}\operatorname{-HORN}}

\newcommand{\ihbs}[1]{{#1}\operatorname{-IHBS}}
\newcommand{\pol}{\operatorname{Pol}}

\newcommand{\eq}{\operatorname{eq}}

\newcommand{\lp}{\operatorname{BLP}}
\newcommand{\lpmin}{\operatorname{BLP}}

\newcommand{\CSP}[1]{\mbox{\rm CSP$(#1)$}}
\newcommand{\MinCSP}[1]{\mbox{\rm Min CSP$(#1)$}}
\newcommand{\MaxCSP}[1]{\mbox{\rm Max CSP$(#1)$}}

\newcommand{\GCSP}[1]{\mbox{\rm GCSP$(#1)$}}
\newcommand{\VCSP}[1]{\mbox{\rm VCSP$(#1)$}}

\newcommand{\PT}{{\mathrm{P}}}
\newcommand{\NP}{{\mathrm{NP}}}
\newcommand{\APX}{{\mathrm{APX}}}

\def\lpVal{\mathsf{Opt}_{\mathsf{LP}}}
\def\opt{\mathsf{Opt}}

\newcommand{\fas}{\phi}
\newcommand{\dist}{\operatorname{dist}}
\def\blpopt{\mathsf{Opt}_{\mathsf{LP}}}

\newtheorem{theorem}{Theorem}
\newdefinition{example}{Example}
\newtheorem{observation}{Observation}
\newtheorem{lemma}{Lemma}

\long\def\BEGINCOMMENT #1\ENDCOMMENT{\relax}

\newcommand{\maps}\longrightarrow
\newcommand{\cmaps}\Longrightarrow

\newcommand{\csp}{\operatorname{CSP}}

\newcommand{\arity}{\rho}

\renewcommand{\arity}{\operatorname{arity}}

\newcommand{\univ}{U}

\newcommand{\Prob}{\Pr}
\newcommand{\loss}{\operatorname{loss}}
\newcommand{\Ind}{\mathbf{1}}
\newcommand{\weight}[1]{w_{#1}}

\usepackage{enumitem}

\def\cspA{A}

\def\cspLang{\Gamma}

\def\lpName{BLP}

\def\lpVal{\mathsf{Opt}_{\mathsf{LP}}}

\def\eqC{\eq_{\cspA}}

\begin{document}

%

\title{Towards a Characterization of Constant-Factor Approximable Finite-Valued CSPs\tnoteref{t1}\tnoteref{t2}}
\tnotetext[t1]{This article is an extended version of a paper published in the proceedings of SODA'15.}
\tnotetext[t2]{The   first   two   authors   were   supported   by   UK   EPSRC   grant
EP/J000078/01.  The first author was also supported by MICCIN grant
TIN2016-76573-C2-1P and Maria de Maeztu Units of Excellence Programme MDM-2015-0502.
The third author was supported by ERC grant 226-203.}
\author[upf]{V\'{\i}ctor Dalmau}
\ead{victor.dalmau@upf.edu}
\cortext[cor1]{Corresponding author}
\author[du]{Andrei Krokhin}
\ead{Andrei.Krokhin@durham.ac.uk}
\author[iitm]{Rajsekar Manokaran}
\ead{rajsekar@gmail.com}

\address[upf]{Department of Information and Communication Technologies, Universitat Pompeu Fabra, Spain}
\address[du]{Department of Computer Science, Durham University, UK}
\address[iitm]{Computer Science and Engineering department, IIT Madras, India}

\begin{abstract}
In this paper we study the approximability of (Finite-)Valued Constraint Satisfaction Problems (VCSPs) with a fixed finite constraint language $\Gamma$ consisting of finitary functions on a fixed finite domain. An instance of VCSP is given by a finite set of variables and a sum of functions belonging to $\Gamma$ and depending on a subset of the variables. Each function takes values in $[0,1]$ specifying costs of assignments of labels to its variables, and the goal is to find an assignment of labels to the variables that minimizes the sum. A recent result of Ene {\it et al.} says that, under the mild technical condition that $\Gamma$ contains the function corresponding to the equality relation, the basic LP relaxation is optimal for constant-factor approximation for $\VCSP\Gamma$ unless the Unique Games Conjecture fails. Using the algebraic approach to the CSP, we give new natural algebraic conditions for the finiteness of the integrality gap for the basic LP relaxation of $\VCSP\Gamma$. We also show how these algebraic conditions can in principle be used to round solutions of the basic LP relaxation, and how this leads to efficient constant-factor approximation algorithms for several examples that cover all previously known cases that are NP-hard to solve to optimality but admit constant-factor approximation. Finally, we show that the absence of another algebraic condition leads to NP-hardness of constant-factor approximation. Thus, our results strongly indicate where the boundary of constant-factor approximability for VCSPs lies.
\end{abstract}






\maketitle

\section{Introduction}
\label{sec:Intro}

The constraint satisfaction problem (CSP) provides a framework in which it is
possible to express, in a natural way, many combinatorial problems encountered in
computer science and AI~\cite{Creignou01:book,Creignou08:complexity,Krokhin17:book}. Standard examples of CSPs include
satisfiability of propositional formulas, graph colouring problems, and systems of linear equations.
An instance of the CSP consists of
a set of variables, a (not necessarily Boolean) domain of labels, and a set
of constraints on combinations of values that can  be taken by certain subsets of variables. The aim is
then to find an assignment of labels to the variables that, in the decision version, satisfies all the
constraints or, in the optimization version, maximizes (minimizes) the number of satisfied (unsatisfied, respectively) constraints.

Since the CSP is NP-hard in full generality, a major line of research in CSP tries to identify
special cases that have desirable algorithmic properties
(see, e.g.~\cite{Creignou01:book,Creignou08:complexity,Krokhin17:book}), the primary motivation being the general picture
rather than specific applications.
The two main ingredients of a constraint are: (a) variables to which it is applied, and (b) relation specifying
the allowed combinations of labels.
Therefore, the main types of restrictions on CSP are: (a) {\em structural} where
the hypergraph formed by sets of variables appearing in individual constraints is restricted~\cite{Gottlob09:tractable,Marx13:tractable}, and (b) {\em language-based} where the constraint language $\Gamma$, i.e. the set of relations that can appear in constraints, is fixed (see, e.g.~\cite{Bulatov??:classifying,Creignou01:book,Feder98:monotone,Krokhin17:book}); the corresponding decision/maximization/minimization problems are denoted by $\CSP\Gamma$, $\MaxCSP\Gamma$, and $\MinCSP\Gamma$, respectively.
The ultimate sort of results in this direction are dichotomy results, pioneered by~\cite{Schaefer78:complexity},
which completely characterise the restrictions with a given desirable property modulo some complexity-theoretic assumptions.
The language-based direction is considerably more active than the structural one, and there are many (partial and full) language-based complexity classification results, e.g.~\cite{Barto14:jacm,Barto16:robust,Bulatov17:dichotomy,Deineko08:constants,Kolmogorov17:complexity,Thapper16:finite,Zhuk17:dichotomy}, but many questions are still open.

Problems Max CSP and Min CSP can be generalised by replacing relations (that specify allowed combinations of labels) with functions that specify
a value in $[0,1]$ (measuring the desirability or the cost, respectively) for each tuple of labels. The goal would then be to find an assignment
of labels that maximizes the total desirability (minimizes the total cost, repectively). The maximization version was studied in~\cite{Brown15:combinatorial,Raghavendra08:optimal} under the name of Generalized CSP, or GCSP, (in fact, functions there can take values in $[-1,1]$), while the minimization version is known as (Finite-)Valued CSP~\cite{Thapper16:finite}. In General-Valued CSP, functions can also take the infinite value to indicate infeasible tuples~\cite{Cohen13:algebraic,Cohen06:soft,Kolmogorov17:complexity}, but we will not consider this case in this paper. In this paper we write VCSP to mean {\em finite-valued} CSP.
We note that~\cite{Ene13:local} write Min CSP to mean what we call VCSP in this paper.
Naturally, both GCSP and VCSP can be parameterized by constraint languages $\Gamma$, now consisting of functions instead of relations.

The CSP has always played an important role in mapping the landscape of approximability of NP-hard optimization problems, see e.g. surveys~\cite{Khot10:UGCsurvey,Makarychev17:survey}.
For example, the famous PCP theorem has an equivalent reformulation in terms of inapproximability of a
certain $\MaxCSP\Gamma$, see~\cite{Arora09:complexity}; moreover, Dinur's combinatorial proof of this theorem~\cite{Dinur07:PCP}
deals entirely with CSPs.
The first optimal inapproximability results~\cite{Hastad01:optimal} by H{\aa}stad were about problems $\MaxCSP\Gamma$,
and they led to the study of a new hardness notion called
approximation resistance (see, e.g.~\cite{Austrin08:thesis,Hastad08:2csp,Khot14:strong}).
The approximability of Boolean CSPs has been thoroughly investigated (see, e.g.~\cite{Agarwal05:approx,Creignou01:book,Guruswami16:mixed,Guruswami12:tight,Hastad01:optimal,Hastad08:2csp,Khot10:UGCsurvey,Khot07:optimal}).
Much work around the Unique Games Conjecture (UGC) directly concerns CSPs~\cite{Khot10:UGCsurvey}. This conjecture states
that, for any $\epsilon> 0$, there is a large enough number $k=k(\epsilon)$ such that it $\NP$-hard to tell $\epsilon$-satisfiable from
$(1-\epsilon)$-satisfiable instances of $\CSP{\Gamma_k}$, where $\Gamma_k$ consists of all graphs of bijections on a $k$-element set.
Many approximation algorithms for classical optimization problems have been shown optimal assuming the UGC~\cite{Khot10:UGCsurvey,Khot07:optimal}.
Raghavendra proved~\cite{Raghavendra08:optimal} that one SDP-based algorithm provides  optimal approximation for all problems $\GCSP\Gamma$ assuming the UGC.
In this paper, we investigate problems $\VCSP\Gamma$ and $\MinCSP\Gamma$ on an arbitrary finite domain that belong to APX, i.e. admit a (polynomial-time) constant-factor approximation algorithm,
proving some results that strongly indicate where the boundary of this property lies.

{\bf Related Work.}
Note that each problem $\MaxCSP\Gamma$ trivially admits a constant-factor approximation algorithm because a random assignment of values to
the variables is guaranteed to satisfy a constant fraction of constraints; this can be derandomized by the
standard method of conditional probabilities. The same also holds for GCSP. Clearly, for $\MinCSP\Gamma$ to admit a constant-factor approximation algorithm, $\CSP\Gamma$
must be polynomial-time solvable.

The approximability of problems $\VCSP\Gamma$ has been studied, mostly for Min CSPs in the Boolean case (i.e., with domain $\{0,1\}$, such CSPs are sometimes called ``generalized satisfiability'' problems), see~\cite{Agarwal05:approx,Creignou01:book}.
We need a few concepts from propositional logic. A clause is {\em Horn} if it contains at most one positive literal, and {\em negative} if it contains only negative literals.
Let $\hornsat{k}$ be the constraint language over the Boolean domain that contains all Horn clauses with at most $k$ variables.
For $k\ge 2$, let $\ihbs{k}$ be the subset of $\hornsat{k}$ that consists of all clauses that are negative or have at most 2 variables.
It is known that, for each $k\ge 2$, $\MinCSP{\ihbs{k}}$ belongs to $\APX$~\cite{Creignou01:book}, and they (and the corresponding dual Horn problems)  are essentially the only such Boolean Min CSPs unless the UGC fails~\cite{Dalmau13:robust}. For $\MinCSP{\hornsat{2}}$, which is identical to $\MinCSP{\ihbs{2}}$, a 2-approximation (LP-based) algorithm is described in~\cite{Guruswami12:tight}, which is optimal assuming the UGC, whereas it is $\NP$-hard to constant-factor approximate $\MinCSP{\hornsat{3}}$~\cite{Guruswami16:mixed}.
If $\neq_2$ is the Boolean relation $\{(0,1),(1,0)\}$, then $\MinCSP{\{\neq_2\}}$ is known as {\sc MinUnCut}. $\MinCSP\Gamma$ where $\Gamma$ consists
of 2-clauses is known as {\sc Min 2CNF Deletion}. The best currently known approximation algorithms for {\sc MinUnCut} and {\sc Min 2CNF Deletion} have approximation ratio $O(\!\sqrt{\log n})$~\cite{Agarwal05:approx}, and it follows from~\cite{Khot07:optimal} that neither problem belongs to $\APX$ unless the UGC is false.
The UGC is known to imply the optimality of the basic LP relaxation for any $\VCSP\Gamma$ such that $\Gamma$ contains the (characteristic function of the) equality relation~\cite{Ene13:local}, extending the line of similar results for natural LP and SDP relaxations for
various optimization CSPs~\cite{Kumar11:strict,ManokaranNRS08,Raghavendra08:optimal}.

An approximation algorithm for any $\VCSP\Gamma$ was also given in the 2013 conference version of~\cite{Ene13:local} (that was claimed to match the LP integrality gap), but its analysis was later found to be faulty and this part was retracted in the 2015 update of~\cite{Ene13:local}.
The SDP rounding algorithm for GCSPs from~\cite{Raghavendra09:how} is discussed in detail in the book~\cite{Gaertner12:approximation}, where it is pointed out that the same algorithm does not work for VCSPs.


Constant-factor approximation algorithms for Min CSP are closely related to certain {\em robust algorithms} for CSP that attracted much attention recently~\cite{Barto16:robust,Dalmau13:robust,Dalmau17:robust,Kun12:robust}. Call an algorithm for $\CSP\Gamma$ {\em robust} if, for every $\epsilon>0$ and every $(1-\epsilon)$-satisfiable
instance of $\CSP\Gamma$ (i.e. at most an $\epsilon$-fraction of constraints can be removed to make the instance satisfiable),
it outputs a $(1-f(\epsilon))$-satisfying assignment (i.e. that fails to satisfy at most a $f(\epsilon)$-fraction of constraints) where $f$ is a function
such that $f(\epsilon)\rightarrow 0$ as $\epsilon\rightarrow 0$ and $f(0)=0$.
CSPs admitting a robust algorithm (with some function $f$) were completely characterised
in~\cite{Barto16:robust}; when such an algorithm exists, one can always choose $f(\epsilon)=O(\log\log{(1/\epsilon)}/\log{(1/\epsilon)})$ for the randomized algorithm and $f(\epsilon)=O(\log\log{(1/\epsilon)}/\sqrt{\log{(1/\epsilon)}})$ for the derandomized version.
A robust algorithm is said to have {\em linear loss} if the function $f$ can be chosen so that $f(\epsilon)=O(\epsilon)$.
The problem of characterizing CSPs that admit a robust algorithm with linear loss was posed in~\cite{Dalmau13:robust}.
It is easy to see that, for any $\Gamma$, $\CSP\Gamma$ admits a robust algorithm with linear loss if and only if $\MinCSP\Gamma$ has a constant-factor approximation algorithm. We will use this fact when referring to results in~\cite{Dalmau13:robust}.

Many complexity classification results for CSP have been made possible by the introduction of the
universal-algebraic approach (see, e.g., survey~\cite{Barto17:polymorphisms}), which extracts
algebraic structure from a given constraint language $\Gamma$ (via operations called {\em polymorphisms} of $\Gamma$)
and uses it to analyze problem instances. This approach was extended to VCSP (see, e.g., survey~\cite{Krokhin17:valued}), where
polymorphisms are replaced by certain probability distributions on operations called {\em fractional polymorphisms}.
The universal-algebraic framework to study robust algorithms with a given loss
was presented in~\cite{Dalmau13:robust}, this approach was also used in~\cite{Barto16:robust,Kun12:robust}.
In this paper, we apply this framework with some old and some new algebraic conditions
to study problems $\VCSP\Gamma$ and $\MinCSP\Gamma$.
Our algebraic conditions use symmetric operations, which appear naturally when LP-based algorithms are used for (V)CSPs; other recent examples are~\cite{Kolmogorov17:complexity,Kolmogorov15:power,Kun12:robust,Thapper16:finite}.

{\bf Contributions.} Some of our results assume that $\Gamma$ contains the equality relation.
We characterise problems $\VCSP\Gamma$ for which the basic LP relaxation has finite integrality gap. The characterisation is in terms of appropriately modified fractional polymorphisms.
We then show how that a description of constant-factor approximable VCSPs can be reduced to that for Min CSPs.
For Min CSPs, we give another algebraic condition that characterizes the property of being constant-factor approximable.
This characterization uses the algebraic approach to CSP that has been extremely fruitful in proving complexity classification results for CSPs.
The characterizing condition is in terms of Lipschitz probability distributions on symmetric polymorphisms of $\Gamma$.
This condition can in principle be used to design efficient constant-factor approximation algorithms, provided
one can efficiently sample from these distributions.
We show that this is possible for some examples that cover all cases where such algorithms (but not algorithms finding an optimal solution) were previously known to exist.

It follows from the \cite{Ene13:local} that every Min CSP for which the basic LP relaxation does not have finite integrality gap is not constant-factor approximable, unless the UGC fails. For a class of Min CSPs we strengthen the UG-hardness to NP-hardness. A near-unanimity polymorphism is a type of polymorphism well known in the algebraic theory of CSP~\cite{Barto12:NU,Barto17:polymorphisms,Feder98:monotone}, and its presence follows from the existence of those Lipschitz distributions. We show $\MinCSP\Gamma$ is NP-hard to constant-factor approximate if $\Gamma$ has no near-unanimity polymorphism.

Thus, our results strongly indicate where the boundary of constant-factor approximability for VCSPs lies.

\section{Preliminaries}

Let $A$ be a finite set. A $k$-tuple $\ta=(a_1,\dots,a_k)$ on $A$ is any element of $A^k$.
A $k$-ary relation on $A$ is a subset of $A^k$.  We shall use $\arity(R)$ to denote the arity of relation $R$.
We shall denote by $\eqC$ the binary relation $\{(a,a) \mid a\in A\}$.

An {\em instance} of the $\csp$ is a triple $I=(V,A,{\mathscr C})$ with $V$ a finite set of {\em variables}, $A$ a finite set called {\em domain}, and ${\mathscr C}$ a finite list of {\em constraints}. Each constraint in ${\mathscr C}$ is a pair $C=(\tv,R)$, also denoted $R(\tv)$, where $\tv=(v_1,\dots,v_k)$ is a tuple of variables of length $k$, called the {\em scope} of $C$, and $R$ an $k$-ary relation on $A$, called the {\em constraint relation} of $C$. The {\em arity} of a constraint $C$, $\arity(C)$, is the arity of its constraint relation. When considering optimization problems, we will assume
that each constraint has a weight $w_C\in \mathbb{Q}_{> 0}$.
It is known (see, e.g.~Lemma~7.2 in~\cite{Creignou01:book}) that allowing weights in $\MinCSP\Gamma$ does not affect membership in $\APX$.

Very often we will say that a constraint $C$ belongs to instance $I$ when, strictly speaking, we should be saying that appears in the constraint list ${\mathscr C}$ of $I$. Also, we  might sometimes write $(v_1,\dots,v_k,R)$ instead of $((v_1,\dots,v_k),R)$.
 A {\em constraint language} is any {\em finite} set $\Gamma$ of relations on $A$. The problem $\csp(\Gamma)$ consists of all instances of the $\csp$ where all the constraint relations are from $\Gamma$.
An assignment for $I$ is a mapping $s:V\rightarrow A$. We say that $s$ {\em satisfies} a constraint $(\tv,R)$ if $s(\tv)\in R$ (where $s$ is applied component-wise).

The {\em decision problem} for  $\csp(\Gamma)$ asks whether an input instance $I$ of $\csp(\Gamma)$ has a solution, i.e., an assignment satisfying all constraints.  The natural {\em optimization problems} for $\csp(\Gamma)$, $\MaxCSP\Gamma$ and $\MinCSP\Gamma$, ask to find an assignment that maximizes the total weight of satisfied constraints or minimizes the total weight of unsatisfied constraints, respectively.

VCSP is the generalization of Min CSP obtained by allowing a richer set of constraints. Formally, a constraint in a VCSP instance is a pair $C=(\tv,\varrho)$, also denoted $\varrho(\tv)$, where $\tv=(v_1,\dots,v_k)$ is, as before, a tuple of variables, and $\varrho:A^k\rightarrow [0,1]$ is a mapping from $A^k$ to $[0,1]$. Given an instance $I$ of the VCSP, the goal is to find
an assignment $s:V\rightarrow A$ that minimizes $$\sum_{C=\varrho(\tv)\in{\mathscr C}} w_C\cdot\varrho(s(\tv)).$$
Note that, to express Min CSP as VCSP, one needs to replace each relation $R$ in a Min CSP instance by a function $\varrho_R$ such that $\varrho_R(\ta)=0$ if $\ta\in R$ and $\varrho_R(\ta)=1$ otherwise.

\subsection{Basic linear program}

Many approximation algorithms for optimization CSPs use the basic (aka standard) linear programming (LP) relaxation~\cite{Dalmau13:robust,Kumar11:strict,Kun12:robust}. 

For any instance $I=(V,A,{\mathscr C})$ of $\VCSP\Gamma$, there is an equivalent canonical 0-1 integer program. It has variables $p_v(a)$ for every $v\in V$, $a\in A$, as well as variables $p_C(\ta)$ for every constraint $C=\varrho(\tv)$ and every tuple $\ta\in A^{\arity(\varrho)}$. The interpretation of $p_v(a)=1$ is that variable $v$ is assigned value $a$; the interpretation of $p_C(\ta)=1$ is that
$\tv$  is assigned (component-wise) tuple $\ta$. More formally, the program ILP is the following:

\begin{align}
& \text{minimize: }  \displaystyle\sum_{C=\varrho(\tv)\in{\mathscr C},\ta\in A^{\arity{\varrho}}} \ \ \ w_C\cdot p_C(\ta)\cdot\varrho(\ta) & & \nonumber \\
&\text{subject to: }  &  &  \nonumber \\
& p_v(A)=1 \hspace{3.8cm} & \mbox{for } v\in V;  & \label{lp1}\\
&   p_C(A^{j-1}\times\{a\}\times A^{\arity(C)-j})=p_{\tv_j}(a) & \mbox{for } C=(\tv,R)\in{\mathscr C}, & \label{lp2} \\
& &  1\leq j\leq \arity(C), a\in A. \nonumber
\end{align}

Here, for every $v\in V$ and $S\subseteq A$, $p_v(S)$ is a shorthand for $\sum_{a\in S} p_v(a)$ and for every $C$ and every $T\subseteq A^{\arity{(C)}}$, $p_C(T)$ is a shorthand for
$\sum_{\ta\in T} p_C(\ta)$.

If we relax this ILP by allowing the variables to take values in the range $[0,1]$ instead of $\{0,1\}$, we obtain the {\em basic linear programming relaxation} for $I$, which we denote by
$\lp(I)$.
As $\Gamma$ is fixed, an optimal solution to $\lp(I)$ can be computed in time polynomial in $|I|$.

For an instance $I$ of $\VCSP\Gamma$, we denote by $\opt(I)$ the value of an optimal solution to $I$, and by $\blpopt(I)$ the value
of an optimal solution to $\lp(I)$.

For any finite set $X$, we shall denote by $\Delta(X)$ the set of all probability distributions on $X$. Furthemore, for any $n\in \mathbb{N}$, we shall denote by $\Delta_{n}(X)$ the subset of $\Delta(X)$ consisting of all $q\in\Delta(X)$ such that $q(x)\cdot n$ is an integer for every $x\in X$. To simplify notation we shall
write $\Delta_{n}$ and $\Delta$ as a shorthand of $\Delta_{n}(A)$ and $\Delta(A)$ respectively.
If $p\in\Delta(A^r)$ and $p_1,\dots,p_r\in\Delta(A)$ will say that {\em the marginals of $p$ are $p_1,\dots,p_r$} to indicate that for every $1\leq i\leq r$, and every $a\in A$, $p(A^{i-1}\times\{a\}\times A^{r-i})=p_i(a)$.

Restriction (\ref{lp1}) of $\lp(I)$ expresses the fact that, for each $v\in V$, $p_v\in \Delta(A)$.  Also, (\ref{lp1}) and ({\ref{lp2}) together express the fact that, for each constraint
$C=\varrho(\tv)$, of arity $k$, we have $p_C\in\Delta(A^{k})$ and that the marginals of the $p_C$ distribution are consistent with the $p_v$ distributions.

Recall that the {\em integrality gap} of BLP for $\VCSP\Gamma$ is defined as
$$
\sup_I{\frac{\opt(I)}{\blpopt(I)}}
$$
where
the supremum is taken over all instances $I$ of $\VCSP\Gamma$. In particular, if the integrality gap is finite, then $\opt(I)=0$
whenever $\blpopt(I)=0$.

Recall that $\eqC$ denotes the binary equality relation. In the following theorem, $\eqC$ will also denote the binary function on $A$ such that
$\eqC(x,y)$ is equal to $0$ if $x=y$ and equal to 1 otherwise. This will not cause any confusion.

\begin{theorem}[\cite{Ene13:local}]
  \label{thm:optLP}%
  Let $\cspLang$ be a constraint language such that
  $\eqC \in \cspLang$ and let $\alpha_{gap}$ be the integality gap of
  \lpName{} for $\VCSP{\cspLang}$.  For every real number $\beta <
  \alpha_{gap}$, it is NP-hard to approximate $\VCSP{\cspLang}$ to within
  a factor $\beta$ unless the UGC is false. In particular, if the integrality
  gap is infinite then there is no constant-factor approximation algorithm for $\VCSP\Gamma$
  unless the UGC is false.
\end{theorem}

The setting in~\cite{Ene13:local} assumes that each variable
in an instance comes with its own list of allowed images (i.e. a subset of $A$), but this assumption is not essential in their reduction from the UGC.

\subsection{Algebra}
\label{sec:algebra}

Most of the terminology introduced in this section is standard.
See~\cite{Bulatov??:classifying,Barto17:polymorphisms} for more detail about the algebraic approach to the CSP.
An $n$-ary {\em operation} on $A$ is a map $f:A^n\rightarrow A$.
Let us now define several types of operations that will be used in this paper.
We usually define operations by identities, i.e. by equations where all variables are assumed to be universally quantified.

\begin{itemize}
\item An operation $f$ is {\em idempotent} if it satisfies the identity $f(x,\dots,x)=x$.
\item An operation $f$ is {\em symmetric} if $f(x_1,\dots,x_n)=f(x_{\pi(1)},\dots,x_{\pi(n)})$ for each permutation $\pi$ on $\{1,\dots,n\}$.

Thus, a symmetric operation is one that depends only on the multiset of its arguments. Since there is an obvious one-to-one
correspondence between $\Delta_n(A)$ and multisets of size $n$,
$n$-ary symmetric operations on $A$ can be naturally identified with functions from $\Delta_n(A)$ to $A$.

\item  An $n$-ary operation $f$ on $A$ is {\em totally symmetric} if $f(x_1,\dots,x_n)=f(y_1,\dots,y_n)$ whenever $\{x_1,\dots,x_n\}=\{y_1,\dots,y_n\}$.
\item An $n$-ary $(n\geq 3)$ operation $f$ on $A$ is called an NU (near-unanimity) operation if it satisfies the identities
$$f(y,x,x\dots,x,x)=f(x,y,x\dots x,x)=\cdots=f(x,x,x\dots x,y)=x.$$

\end{itemize}


An $n$-ary operation $f$ on $A$ {\em preserves} (or is a {\em polymorphism} of) a $k$-ary relation $R$ on $A$ if for every $n$ (not necessarily distinct) tuples
$(a^i_1,\dots,a^i_k)\in R$, $1\leq i\leq n$, the tuple
$$(f(a^1_1,\dots,a^n_1),\dots,f(a^1_k,\dots,a^n_k))$$
belongs to $R$ as well. Given a set $\Gamma$ of relations on $A$, we denote by $\pol(\Gamma)$ the set of all operations $f$ such that $f$ preserves each relation in $\Gamma$. If $f\in \pol(\Gamma)$ then $\Gamma$ is said to be {\em invariant} under $f$. If $R$ is a relation we might freely
write  $\pol(R)$ to denote $\pol(\{R\})$.

\begin{example}\label{ex:poly}
Let $A=\{0,1\}$.
\begin{enumerate}
\item It is well known and easy to check that, for each $n\ge 1$, the $n$-ary (totally symmetric) operation $f_n(x_1,\dots,x_n)=\bigwedge_{i=1}^n{x_i}$ is a polymorphism of $\hornsat{3}$.
\item
It is well known and easy to check that, for each $k\ge 2$, constraint language $\ihbs{k}$, as defined in Section~\ref{sec:Intro}, has polymorphism $x\wedge (y\vee z)$, but the operation $x\vee y$ is not a polymorphism of $\ihbs{k}$.
\end{enumerate}
\end{example}


In general, it is well known that the set $\pol(\Gamma)$ of any constraint language $\Gamma$ is a clone, which means that
it contains all projections (i.e. operations of the form $p_n^i(x_1,\ldots,x_n)=x_i$) and is closed
under composition. The latter condition is spelled out as follows: if $f,g_1,\dots,g_n$ are polymorphisms of $\Gamma$ where $f$ is $n$-ary and $g_1,\dots,g_n$ are $m$-ary
then the $m$-ary operation $h(x_1,\dots,x_m)=f(g_1(x_1,\dots,x_m),\dots,g_n(x_1,\dots,x_m))$ is also a polymorphism of $\Gamma$.

The complexity of constant-factor approximation for $\MinCSP\Gamma$ is completely determined by $\pol(\Gamma)$, as the next theorem indicates.

\begin{theorem}[\cite{Dalmau13:robust}]
\label{the:Galois}
Let $\Gamma$ and $\Gamma'$ be constraint languages on $A$ such that $\pol(\Gamma) \subseteq \pol({\Gamma'})$. Assume, in addition, that $\Gamma$ contains the equality relation $\eqC$.
Then, if $\MinCSP\Gamma$ has a constant-factor approximation algorithm then so does $\MinCSP{\Gamma'}$.
\end{theorem}

We say that BLP {\em decides} $\CSP\Gamma$ if, for any instance $I$ of $\CSP\Gamma$, $I$ is satisfiable whenever $\lpVal(I)=0$.

\begin{theorem}[\cite{Kun12:robust}]
\label{LP-sym}
For any $\Gamma$, the following are equivalent:
\begin{enumerate}
\item $\lp$ decides $\CSP\Gamma$,
\item $\Gamma$ has symmetric polymorphisms of all arities.
\end{enumerate}
\end{theorem}

Note that symmetric polymorphisms provide a natural rounding for $\lp(I)$, as follows. Let $s$ be an optimal solution to $\lp(I)$ in which all variables are assigned rational numbers such that, for some $n\in \mathbb{N}$, $p_v\in\Delta_n(A)$ for each variable $v$ in $I$ and $p_C\in\Delta_n(A^{arity(C)})$ for each constraint $C$ in $I$. Then each $v$ can be assigned the element $f(p_v)$
where $f$ is any fixed $n$-ary symmetric polymorphism of $\Gamma$. It is not hard to check (or see~\cite{Kun12:robust}) that if $\blpopt(I)=0$ then this assignment will satisfy all constraints in $I$.

It was claimed in~\cite{Kun12:robust} that the conditions in Theorem~\ref{LP-sym} are also equivalent to the condition of having totally symmetric polymorphisms of all arities, but a flaw was later discovered in the proof of this claim, and indeed a counterexample (see Section~\ref{sec:algorithms}) was found in~\cite{Kun16} (Example 99).

\section{Results}

We will first formally state our main results and then go into detailed proofs.

\subsection{Formal statements of main results}

For any function $\varrho:A^k\rightarrow [0,1]$, let $R_\varrho$ denote the $k$-ary relation $R_\varrho=\{\ta \mid \varrho(\ta)=0\}$.

\begin{theorem}\label{thm:reduction}
Let $\Gamma_1$ be a valued constraint language and let $\Gamma_2=\{R_\varrho \mid \varrho\in \Gamma_1\}$.
Then $\VCSP{\Gamma_1}$ is in $\APX$ if and only if $\MinCSP{\Gamma_2}$ is in $\APX$.
\end{theorem}

Hence, for every valued constraint language there is an equivalent (relational) constraint language. Due to this reduction, we can freely focus on Min CSPs. Regarding Min CSPs, we will formulate most of our results for constraint languages $\Gamma$ that contain the equality relation $\eqC$.
We make this restriction because some of the reductions in this paper and some papers that we use are currently known to work only with this restriction. We conjecture that this restriction is not essential, that is, for any $\Gamma$, $\MinCSP\Gamma$ admits a constant-factor approximation algorithm if and only if $\MinCSP{\Gamma\cup\{\eqC\}}$ does so (though the optimal constants may differ).



As mentioned before, for any $\Gamma$, $\CSP\Gamma$ admits a robust algorithm with linear loss if and only if $\MinCSP\Gamma$ has a constant-factor approximation algorithm. For constraint languages $\Gamma$ containing the equality relation $\eqC$, it follows from results in Section~3 of~\cite{Dalmau13:robust} that a complete characterisation of constant-factor approximability of Min CSP reduces to the case
when $\Gamma$ contains all unary singletons, i.e., relations $\{a\}$, $a\in A$. Hence, some statements in the paper will (explicitly) assume this condition. Note that this condition implies that all polymorphisms of $\Gamma$ are idempotent.

Theorem~\ref{thm:optLP} provides evidence, in terms of integrality gap, that the \lpName{} is optimal to design constant-factor
approximation algorithms for $\MinCSP\Gamma$. Our main result is a characterization of problems
$\MinCSP\Gamma$ for which BLP has a finite integrality gap.

For $p,q\in \Delta$, let $\dist(p,q)=\max_{a\in A}{|p(a)-q(a)|}$.
For a tuple $\mathbf{a}\in A^n$, let $d_{\mathbf{a}}\in \Delta_n$ be such that each element $x\in A$ appears in $\mathbf{a}$
exactly $n\cdot d_{\mathbf{a}}(x)$ times. For tuples $\mathbf{a},\mathbf{b}\in A^n$, define $\dist(\mathbf{a},\mathbf{b})=\dist(d_\mathbf{a},d_\mathbf{b})$.
An $n$-ary {\em fractional operation} $\fas$ on $A$ is any probability distribution on the set of $n$-ary operations on $A$.
For every real number $c\geq 0$, call $\fas$ {\em $c$-Lipschitz\footnote{In the conference version we used the terminology {\em stable} instead of {\em Lipschitz}.}} if its support consists of symmetric operations and, for all $\mathbf{a},\mathbf{b}\in A^{n}$, we have
$\Pr_{g\sim \fas}\{g(\mathbf{a})\ne g(\mathbf{b})\}\le c\cdot \dist(\mathbf{a},\mathbf{b})$.

\begin{theorem}
\label{the:Lipschitz}
For any $\Gamma$ containing $\eqC$, the following are equivalent:
\begin{enumerate}
\item The integrality gap of BLP for $\MinCSP\Gamma$ is finite.
\item There is $c\ge 0$ such that, for each $n\in\mathbb{N}$, there is an $n$-ary $c$-Lipschitz fractional operation $\fas_n$ on $A$ whose support consists of symmetric polymorphisms of $\Gamma$.
\end{enumerate}
\end{theorem}

We now give an example of how Theorem~\ref{the:Lipschitz} can be applied to prove negative results. Recall Example~\ref{ex:poly}.  It is known and not hard to check that the operation $f_n$ is the only $n$-ary symmetric polymorphism of $\hornsat{3}$. It follows that there is only one fractional operation of arity $n$, $\fas_n$, whose support consists of symmetric polymorphisms
and that $\Pr_{g\sim \fas_n}\{g(\mathbf{a})\ne g(\mathbf{b})\}=1$ if we choose $\mathbf{a}=(1,1,\ldots,1)$ and $\mathbf{b}=(0,1,\ldots,1)$. Since $\dist(\mathbf{a},\mathbf{b})=1/n$,  it follows that there is no constant $c\ge 0$ satisfying condition (2) of Theorem~\ref{the:Lipschitz}, and hence the integrality gap of $\lp$ for $\MinCSP{\hornsat{3}}$ is infinite (and constant-factor approximation for $\MinCSP{\hornsat{3}}$ is UG-hard).

On the algorithmic side, any sequence $\fas_n$, $n\in\mathbb{N}$, satisfying condition (2) of Theorem \ref{the:Lipschitz} can be used to obtain a
(possibly efficient) randomized rounding procedure for BLP, as follows. As we explained after Theorem~\ref{LP-sym},
if one has an optimal rational solution to $\lp(I)$, one
can use a symmetric operation of appropriate arity $n$ to round this solution to obtain a solution for $I$.
If the symmetric operation is drawn from a $c$-Lipschitz distribution $\fas_n$ on $n$-ary symmetric polymorphisms (such as in Theorem~\ref{the:Lipschitz})
then this procedure is a randomized constant-factor approximation algorithm for $\MinCSP\Gamma$ (this follows from the proof of direction $(2)\Rightarrow (1)$ of Theorem~\ref{the:Lipschitz}).
However it is not entirely clear how to efficiently sample from $\fas_{n}$. In Subsection~\ref{sec:algorithms}, we give
two examples - a class of constraint languages and one specific language - with sequences of Lipschitz distributions that are nice enough to admit efficiently sampling.
The first of these examples, given in Theorem~\ref{the:algorithm}, covers (in a specific sense - see discussion in Subsection~\ref{sec:algorithms}) all problems $\MinCSP\Gamma$ that were previously known to belong to $\APX$, but are not efficiently solvable to optimality.

\begin{theorem}
\label{the:algorithm}
Let $A$ consist of subsets of a set and suppose that $A$ is closed under intersection $\cap$ and union $\cup$. If a constraint language $\Gamma$ on $A$ has polymorphism $x\cap (y\cup z)$
then $\MinCSP\Gamma$ has a constant-factor approximation algorithm. \end{theorem}

Our last result strengthens UG-hardness to NP-hardness for a class of Min CSPs. A near-unanimity polymorphism (see definition in Subsection~\ref{sec:algebra}) is a type of polymorphism well known in the algebraic theory of CSP~\cite{Barto12:NU,Barto17:polymorphisms,Feder98:monotone}, and its presence is implied by the existence of those Lipschitz distributions (see Subsection~\ref{sec:NU}). We show $\MinCSP\Gamma$ is NP-hard to constant-factor approximate if $\Gamma$ has no near-unanimity polymorphism.

\begin{theorem}
\label{the:NU}
Let $\Gamma$ be a constraint language containing $\eqC$ and all unary singleton relations. If $\MinCSP\Gamma$ admits a constant-factor approximation algorithm then $\Gamma$ has an NU polymorphism, unless $\PT=\NP$.
\end{theorem}

It is well known and easy to check that $\hornsat{3}$ has no NU polymorphism of any arity, so, by Theorem~\ref{the:NU}, $\MinCSP{\hornsat{3}}$ is NP-hard to constant-factor approximate.

\subsection{Reduction from VCSP to Min CSP}

%

\begin{proof} (of Theorem~\ref{thm:reduction}).
We can assume that some function in $\Gamma_1$ takes a positive value, since otherwise both problems are trivial.
Let $m>0$ denote the minimal positive value taken by any function in $\Gamma_1$. Note, that there is a natural one-to-one correspondence between instances in $\VCSP{\Gamma_1}$ and instances in $\MinCSP{\Gamma_2}$, namely, every instance $I_1=(V,A,{\mathcal C})$ in $\VCSP{\Gamma_1}$ is associated to the instance $I_2$ in $\MinCSP{\Gamma_2}$, obtained by replacing every constraint $(\tv,R_\varrho)$ in $I_1$ by $\varrho(\tv)$. The theorem follows from the observation that the values every assignment $s:V\rightarrow A$ in $I_1$ and $I_2$ are within a multiplicative factor of each other. More precisely, if $v_1$ and $v_2$ are the values of assignment $s$ for intances $I_1$ and $I_2$ respectively, then
$$v_1\leq v_2\leq \frac{v_1}{m}.$$
\end{proof}

\subsection{Finite integrality gaps}

In this subsection we prove Theorem~\ref{the:Lipschitz}. We need a few definitions and intermediate results.

Let $I$ be any weighted instance in $\MinCSP\Gamma$ with variable set $V$. A {\em fractional assignment} for $I$ is any probability distribution, $\fas$, on the set of assignments for $I$. For a real number $c\geq 1$, we say that a fractional assignment $\fas$ for $I$ is {\em $c$-bounded} if, for every constraint $C=(v_1,\dots,v_r,R)$ in $I$,
$$\Pr_{g\sim\fas}\{g(v_1),\dots,g(v_r))\not\in R\}\leq c\cdot(1-\weight{C})$$
where $\weight{C}$ is the weight in $I$ of constraint $C$. We will apply it only to instances where $w_C\in [0,1]$.

For every relation $R\in\Gamma$ of arity, say $r$, and every $p_1,\dots,p_r\in\Delta$ we define $\loss(p_1,\dots,p_r,R)\in[0,1]$ to be $\min_p (1-p(R))$ where $p$ ranges over all the probability distributions on $A^r$ with marginals $p_1,\dots,p_r$.

In a technical sense, the function $\loss$ 'encodes' the contribution of each constraint in optimal solutions of BLP. This is formalized in the following observation.

\begin{observation}
\label{obs:universal}
Let $I$ be any instance of $\MinCSP\Gamma$ and let $C=(v_1,\dots,v_r,R)$ be any of its constraints. Then $1-p_C(R)=\loss(p_{v_1},\dots,p_{v_r},R)$ holds in any optimal solution of $\lpmin(I)$.
\end{observation}

For every $n\in\mathbb{N}$, the $n$-th {\em universal instance} for $\Gamma$, $\univ_{n}(\Gamma)$, is the instance with variable set $\Delta_{n}$ containing for every relation $R$ of arity, say $r$, in $\Gamma$, and every $p_1,\dots,p_r\in \Delta_{n}$, constraint $(p_1,\dots,p_r,R)$ with weight $1-\loss(p_1,\dots,p_r,R)$. We write simply $\univ_{n}$ if $\Gamma$ is clear from the context.


The following is a variant of Farkas' lemma 
that we will use in our proofs.

\begin{lemma}(Farkas' Lemma)
\label{farkas}
Let $M$ be a $m\times n$ matrix, $b\in \mathbb{R}^m$. Then exactly one of the following two statements is true:
\begin{enumerate}
\item There is an $x\in (\mathbb{R}_{\geq 0})^n$ with $\|x\|_1=1$ ($\|x\|_1$ denotes the $1$-norm of $x$) such that $Mx\leq b$.
\item There is a $y\in (\mathbb{R}_{\geq 0})^m$ with $\|y\|_1=1$ such that $y^Tb<y^TM$ (i.e. each coordinate of $y^TM$ is strictly greater than $y^Tb$).
\end{enumerate}
\end{lemma}
\begin{proof}
Condition (1) is equivalent to the existence of $x\in (\mathbb{R}_{\geq 0})^n$ such that $M'x\leq b'$ where
$M'$ and $b'$ are obtained by adding two extra rows to $M$ and $b$ expressing that
that $\sum_{1\leq i\leq n} x_i\leq 1$ and $\sum_{1\leq i\leq n} -x_i\leq -1$.  It then follows from Corollary 7.1f in~\cite{Schrijver86:book} that the negation of condition (1) is equivalent to the existence of a
vector $z\in (\mathbb{R}_{\geq 0})^{(m+2)}$ satisfying $z^TM'\geq 0$ and $z^T b'<0$. It is easy to see that this is equivalent to condition (2).
\end{proof}

Theorem \ref{the:Lipschitz} follows directly from Lemmas \ref{le:bounded} and \ref{le:Lipschitz} below.

\begin{lemma}
\label{le:bounded}
For every constraint language $\Gamma$ and $c\geq 1$, the following are equivalent:
\begin{enumerate}
\item The integrality gap of BLP for $\MinCSP\Gamma$ is at most $c$.
\item For each $n\in\mathbb{N}$, there is a $c$-bounded fractional assignment for $\univ_{n}$.
\end{enumerate}
\end{lemma}
\begin{proof}
This proof is an adaptation of the proof of Theorem 1 in \cite{Kolmogorov15:power}, and it also works for valued CSPs.

$(2\Rightarrow 1)$ Let $I=(V,A,{\mathscr C})$ be any instance of $\MinCSP\Gamma$, and let $p_v(v\in V)$, $p_C(C\in{\mathscr C})$ by any optimal solution of $\lpmin(I)$. We can assume that there exists $n\in\mathbb{N}$ such that $p_v\in\Delta_{n}$ for every $v\in V$. For every assignment $g$ for $\univ_{n}$, let $s_g$ be the assignment for $I$ defined as $s_g(v)=g(p_v), v\in V$.

Since (2) holds, it follows from Observation \ref{obs:universal} and the definition of $c$-boundedness that, for every constraint
$C=(v_1,\dots,v_r,R)$ in $I$, we have
$$\Pr_{g\sim\fas}\{(s_g(v_1),\dots,s_g(v_r))\not\in R\}\leq c\cdot (1-p_C(R))$$
It follows that the expected value of $s_g$ is at most $c\cdot \blpopt(I)$. Consequently, there exists some $s_g$ with value at most $c\cdot \blpopt(I)$.

$(1\Rightarrow 2)$ We shall prove the contrapositive. Assume that for some $n\in\mathbb{N}$, there is no $c$-bounded fractional assignment for $\univ_{n}$. We shall write a system of linear inequalities that expresses the existence of a $c$-bounded fractional assignment for $\univ_{n}$ and then apply Lemma~\ref{farkas} to this system. To this end, we introduce a variable $x_g$ for every assignment $g$ for $\univ_{n}$. The system contains, for every constraint $C=(p_1,\dots,p_r,R)$ in $\univ_{n}$, the inequality:

$$\sum_{g\in G_{n}}  x_g \cdot \Ind[(g(p_1),\dots,g(p_r))\not\in R]\leq c\cdot \loss(C)$$
where $G_{n}$ is the set of all assignments for $\univ_{n}$ and $\Ind[(g(p_1),\dots,g(p_r))\not\in R]$ is $1$ if $g(p_1),\dots,g(p_r))\not\in R$
and $0$ otherwise.
Note that the system does not include equations for $x_g\geq 0$ and $\sum_{g\in G_{n}} x_g=1$ since this is already built-in in the version of Farkas' lemma that we use.

Since there is no $c$-bounded fractional assignment for $\univ_{n}$ it follows from Farkas' Lemma that the system containing for every $g\in G_{n}$ inequality
\begin{equation}\label{eq:ineq}
\sum_{C=(p_1,\dots,p_r,R)\in\univ_{n}} y_C \cdot c\cdot \loss(C)<\sum_{C=R(p_1,\dots,p_r,R)\in\univ_{n}} y_C \cdot \Ind[(g(p_1),\dots,g(p_r))\not\in R]
\end{equation}

has a solution where every variable $y_C$ takes non-negative values and it holds that $\sum_{C} y_C=1$. We can also assume the value of every variable in the solution is rational, since so are all the coefficients in the system.

Now consider instance $I=(V,A,{\mathscr C})$ where $V=\Delta_{n}$ and ${\mathscr C}$ contains, for every relation $R\in\Gamma$ of arity, say $r$, and every $p_1,\dots,p_r\in \Delta_{n}$, constraint $C=(v_1,\dots,v_r,R)$ with weight $y_C$.

We shall construct a solution $p_v (v\in\Delta_{n}), p_C (C\in{\mathscr C})$ of $\lpmin(I)$. For every $v\in\Delta_{n}$, set $p_v$ to $v$ (note that $v$ is a distribution on $A$). For every $C\in{\mathscr C}$ set $p_C$ to the distribution $q$ with $1-q(C)=\loss(C)$. Hence, the objective value of the solution of $\lpmin(I)$ thus constructed
is  $\sum_{C\in\univ_{n}} y_C \cdot  \loss(C)$, which is $c$ times smaller than the left side of inequality \eqref{eq:ineq}. Furthermore, the total weight of falsified constrains by any assignment $g$ for $I$ is precisely the right side of inequality \eqref{eq:ineq}. It follows that the gap of instance $I$ is larger than $c$.
\end{proof}




For every set $X$, one can associate to every $p\in\Delta_{n}(X)$  the multiset $p'$ such every element $x\in X$ occurs in $p'$ exactly $p(x)\cdot n$ times. In a similar way, one obtains a one-to-one correspondence between the assignments (resp. fractional assignments) for $\univ_{n}$ and the $n$-ary symmetric operations (resp. fractional operations with support consisting of $n$-ary symmetric operations).

\begin{lemma}
\label{le:Lipschitz}
For every constraint language $\Gamma$ containing the equality relation $\eqC$, the following are equivalent:
\begin{enumerate}
\item There is $c\geq 1$, such that for each $n\in\mathbb{N}$, there is a $c$-bounded fractional assignment for $\univ_{n}(\Gamma)$.
\item There is $d\geq 0$ such that, for each $n\in\mathbb{N}$, there is an $n$-ary $d$-Lipschitz fractional operation on $A$ whose support consists of symmetric polymorphisms of $\Gamma$.
\end{enumerate}
\end{lemma}
\begin{proof}
\newcommand{\assoc}{\approx}
In this proof it is convenient to distinguish formally between a distribution $y$ (resp. assignment, fractional assignment)
and its associated multiset $y$ (resp. operation, fractional operation)  that, whenever $X$ and $n$ are clear from the context, we shall denote by $y'$. The following observation will be  useful.
\begin{observation}
\label{obs:multiset}
For any assignment $g$ for $\univ_{n}$ and any distribution $p\in \Delta_{n}(A^r)$ we have that
$(g(p_1),\dots,g(p_r))=g'(p')$ where $p_1,\dots,p_r\in \Delta_n(A)$ are the marginals of $p$ and $g'(p')$ denotes the $r$-ary tuple obtained by applying the symmetric $n$-ary operation $g'$ (corresponding to $g$) to the $n$ tuples in $p'$ component-wise.
\end{observation}

$(1)\Rightarrow(2)$ Assume that $\fas$ is a $c$-bounded fractional assignment for $\univ_{n}$. We claim that for every mapping $g$ in the support of $\fas$, $g'$ is, in fact, a polymorphism of $\Gamma$. Indeed, let $R$ be any relation of arity, say $r$, in $\Gamma$, let $t_1,\dots,t_{n}\in R$. We want to show that $g'(t_1,\dots,t_n)\in R$ where
$g'(t_1,\dots,t_n)$ denotes the $r$-ary tuple obtained by applying $g'$ to $t_1,\dots,t_n$ component-wise.

Let $p\in\Delta_{n}(A^r)$ be the distribution associated to
multiset $p'=[t_1,\dots,t_{n}]$ and consider constraint $C=(p_1,\dots,p_r,R)$ on $\univ_{n}$ where $p_1,\dots,p_r$ are the marginals of $p$. By the choice of $p$ we have $p(R)=1$. Since $\fas$ is $c$-bounded it follows that $\Prob_{g\sim\fas}\{(g(p_1),\dots,g(p_r))\not\in R\}\leq c\cdot \loss(C)\leq 1-p(R)=0$. Hence, $(g(p_1),\dots,g(p_r))\in R$ for every $g$ in the support of $\fas$. It follows from Observation~\ref{obs:multiset} that $g'(t_1,\dots,t_n)=(g(p_1),\dots,g(p_r))$ and we are done.

We have just seen that the support of the fractional $n$-ary operation, $\fas'$, associated to $\fas$ consists of polymorphisms of $\Gamma$. Since, by definition, the support of $\fas'$ only contains symmetric operations, in order to complete the proof it suffices to show that $\fas'$ is  $(c\cdot |A|)$-Lipschitz.

Let $p'_1,p'_2\in A^{n}$ and consider constraint $C=(p_1,p_2,\eqC)$ in $\univ_{n}$ where $p_1,p_2\in\Delta_{n}$ are the distributions associated to $p'_1$ and $p'_2$ respectively and $\eqC$ is the equality relation on $A$. It is not too difficult to find a distribution $p$ on $A^2$ with marginals $p_1$ and $p_2$ such that $1-p(\eqC)\leq |A|\cdot \dist(p'_1,p'_2)$. A concrete example can be obtained as follows. For every $a\in A$, let $a_1=\max\{p_1(a)-p_2(a),0\}$,  and $a_2=\max\{p_2(a)-p_1(a),0\}$. Also, let $s=\sum_a a_1=\sum_a a_2$. Then we define $p$ as follows:
$$p(a,b)= \begin{cases} \min\{p_1(a),p_2(b)\} & \mbox{ if } a=b \\
                       \frac{a_1\cdot b_2}{s} & \mbox{ if } a\neq b
\end{cases}$$
It is easy to verify that $p$ satisfies the desired conditions.
Finally, we have
$$\Pr_{g'\sim\fas'}\{g'(p'_1)\neq g'(p'_2)\}=\Pr_{g\sim\fas}\{(g(p_1),g(p_2))\not\in\eqC\}\leq c\cdot\loss(C)\leq c\cdot |A|\cdot \dist(p'_1,p'_2).$$

We note that this is the only part where the condition $\eqC\in\Gamma$ is required.

$(2)\Rightarrow(1)$.
For every $n\in\mathbb{N}$, let $n'$ be a multiple of $n$ to be fixed later,  let $\fas'$ be a $d$-Lipschitz fractional operation of arity $n'$ whose support consists of symmetric polymorphisms of $\Gamma$, and let $\fas$ be its associated fractional assignment for $\univ_{n'}$. We can without loss of generality assume that $d\ge 1$. We shall prove later that, for every constraint
$C=(p_1,\dots,p_r,R)$ in $\univ_{n}$ (note, not in $\univ_{n'}$), we have
\begin{equation}\label{eq:stob}\Pr_{g\sim\fas}\{(g(p_1),\dots,g(p_r))\not\in R\}\leq 2\cdot r\cdot d\cdot \loss(C)\end{equation}

Consider now the fractional assignment $\fas^*$ on $\univ_{n}$ where for every assignment $f$ on $\univ_{n}$,  $\fas^*(f)=\sum_g \fas(g)$ where $g$ ranges over all assignments for $\univ_{n'}$ that {\em extend f} (that is, such that
$f(p)=g(p)$ for every $p\in\Delta_{n}$). It follows from the definition $\fas^*$
that
$$\Pr_{f\sim\fas^*}\{(f(p_1),\dots,f(p_r))\not\in R\}=\Pr_{g\sim\fas}\{(g(p_1),\dots,g(p_r))\not\in R\}$$
for every constraint $(p_1,\dots,p_r,R)$ in $\univ_{n}$. This gives a way to construct, for every $n\in\mathbb{N}$, a $(2\cdot K\cdot d)$-bounded fractional assignment for $\univ_{n}$
where $K$ is the maximum arity of a relation in $\Gamma$.

To finish the proof it only remains to prove inequality \eqref{eq:stob} for any constraint $C=(p_1,\dots,p_r,R)$ in $\univ_{n}$.  Let $p$ be a distribution on $A^r$ such that $1-p(R)=\loss(C)$ is achieved. We can assume that $\loss(C)\leq 1/2$ since otherwise there is nothing to prove.

Note that we can assume that $p(t)$ is rational for every $t\in A^r$. Let $q$ be the distribution on $A^r$ defined as
$$q(t)=\begin{cases}
p(t)/p(R) & t\in R \\
0                   & t\not\in R
\end{cases}$$
Consider constraint $(q_1,\dots,q_r,R)$ where $q_1,\dots,q_r$ are the marginals of $q$. Since the number of constraints in $\univ_{n}$ is finite we can assume that
$n'$ has been picked such that  $q\in \Delta_{n'}(A^r)$. We claim that  $(g(q_1),\dots,g(q_r))\in R$ for any $g$ in the support of $\fas$. Indeed, if $q'=[t_1,\dots,t_{n'}]$ is the multiset of tuples in $A^r$ associated to $q$ then by Observation \ref{obs:multiset} $(g(q_1),\dots,g(q_r))=g'(t_1,\dots,t_{n'})$ and the latter tuple belongs to
$R$ because $g'$ is a polymorphism of $\Gamma$.

We claim that $\dist(p_i,q_i)\leq 2\cdot \loss(C)$
for every $1\leq i\leq r$.
By definition, for every $a\in A$, $q_i(a)=\sum_{t=(t^1,\ldots,t^r)\in R,\ t^i=a}{p(t)/p(R)}$.
Since $p(R)=1-\loss(C)$, we have
\[
\frac{p_i(a)-\loss(C)}{1-\loss(C)}=\sum_{t^i=a}{\frac{p(t)}{p(R)}}-\sum_{t\not\in R}{\frac{p(t)}{p(R)}}\le
q_i(a)\le \sum_{t^i=a}{\frac{p(t)}{p(R)}}=\frac{p_i(a)}{1-\loss(C)}
\]
for every $a\in A$. Moreover, we have $\frac{p_i(a)}{1-\loss(C)}\leq p_i(a)+2\loss(C)$, which immediately follows from
$\frac{p_i(a)}{1-\loss(C)}=p_i(a)+\frac{\loss(C) p_i(a)}{1-\loss(C)}$, $\loss(C)\leq 1/2$ and
$p_i(a)\leq 1$. Hence, $q_i(a)\in \left[p_i(a)-\loss(C),p_i(a)+2\cdot \loss(C)\right]$ for every $a\in A$.
We conclude that
\[
\Prob_{g\sim\fas}\{(g(p_1),\dots,g(p_r))\not\in R\}
\leq \Prob_{g\sim\fas}\{\exists i \mbox{ such that } g(p_i)\neq g(q_i)\}\leq 2\cdot r\cdot d\cdot \loss (C).
\]
\end{proof}

\subsection{Algorithms}
\label{sec:algorithms}

\newcommand{\cod}{\operatorname{ind}}

We now prove Theorem~\ref{the:algorithm} and then describe another constraint language $\Gamma$ which admits nicely structured Lipschitz distributions on symmetric polymorphisms (so that its Min CSP is constant-factor approximable).

Two classes of CSPs were introduced and studied in~\cite{Carvalho11:lattice}, one is a subclass of the other.
We need two notions to define these classes. A {\em distributive lattice} $(L,\cap,\cup)$ is a (lattice representable by a) family $L$ of subsets of a set
closed under intersection $\cap$ and union $\cup$. We say that two constraint languages $\Gamma_1=\{R_1^{(1)},\dots,R_m^{(1)}\}$ on domain $A$ and $\Gamma_2=\{R_1^{(2)},\dots,R_m^{(2)}\}$ on domain $B$, where the arities of corresponding relations match, are {\em homomorphically equivalent} if there are two mappings $f:A\rightarrow B$, $g:B\rightarrow A$ such that for all $1\leq i\leq m$,
$f(t_1)\in R_i^{(2)}$ for every $t_1\in R_i^{(1)}$ and $g(t_2)\in R_i^{(1)}$ for every $t_2\in R_i^{(2)}$.
The smaller class, which we shall call ${\mathcal C}$ (from 'languages with caterpillar duality'), consists of constraint languages $\Gamma$ such that $\Gamma$ is homomorphically equivalent to a constraint language $\Gamma'$
on some family $L$ of subsets of a finite sets that has polymorphisms $\cap$ and $\cup$, where $\cap$ and $\cup$ are the usual set-theoretic union and intersection (i.e. $(L,\cap,\cup)$ is a finite distributive lattice).
The larger class, which we shall call $\mathcal{J}$ (from 'languages with jellfish duality') is defined similarly, but we require $\Gamma'$ to have polymorphism $x\cap (y\cup z)$.
Constraint languages $\ihbs{k}$ (defined in Section~\ref{sec:Intro}) belong to the class $\mathcal{J}$, but not to $\mathcal{C}$ (see Example~\ref{ex:poly}).
See~\cite{Carvalho11:lattice} for other specific examples of CSPs contained in these classes.
For every $\Gamma$ in $\mathcal{C}$, $\MinCSP\Gamma$ was shown to belong to $\APX$ in~\cite{Kun12:robust}. This result was extended to $\mathcal{J}$ in~\cite{Dalmau13:robust} (see Theorems~5.8 and 4.8 there).


We will now show how Lipschitz distributions on symmetric polymorphisms can be used to provide a constant-factor approximation algorithm for $\MinCSP\Gamma$ for every $\Gamma$ in this class. Observe that if $\Gamma$ and $\Gamma'$ are homomorphically equivalent then $\MinCSP\Gamma$ and $\MinCSP{\Gamma'}$ are essentially the same problem because there is an obvious one-to-one correspondence between
instances of $\MinCSP{\Gamma_1}$ and $\MinCSP{\Gamma_2}$ (swapping $R_i^{(1)}$ and $R_i^{(2)}$ in all constraints) and
the maps $f$ and $g$ allow one to move between solutions to corresponding instances without any loss of quality.
So, we can assume that $A$ consists of subsets of some set, and $\Gamma$ has polymorphism $x\cap (y\cup z)$ where $(A,\cap,\cup)$ is a distributive lattice.



Throughout the section, $K$ will denote the maximum arity of a relation in such $\Gamma$. For every $1\leq h\leq n$, let $g_{h,n}(x_1,\dots,x_{n})$ be the $n$-ary symmetric operation on $A$ defined as
$$\bigcup_{I\subseteq\{1,\dots,n\}, |I|= h} \left(\bigcap_{i\in I} x_i\right)$$

\begin{lemma}
\label{le:jellypoly}
For all $h,n\in\mathbb{N}$ with $\left(1-\frac{1}{|A|^K}\right)n<h\leq n$, we have $g_{h,n}\in\pol(\Gamma)$.
\end{lemma}
\begin{proof}
It is not difficult to see that $x\cap y$ is also a polymorphism of $\Gamma$. Indeed, for every relation $R$ and every pair of tuples $t,t'\in R$, we have that $t\cap t'=t\cap (t'\cup t')$ and hence it belongs to $R$.
Using composition, we shall show that $R$ has polymorphism $f_{h,n}$ where
$f_{h,n}(x_0,x_1,\dots,x_{n})$ is the $(1+n)$-ary operation defined as
$$x_0\cap g_{h,n}(x_1,\dots,x_{n})=x_0\cap\left (\bigcup_{I\subseteq\{1,\dots, n\}, |I|=h} \left(\bigcap_{i\in I} x_i\right)\right)$$

First, we observe that for every $m\geq 2$, the $m$-ary operation $x_1 \cap\cdots\cap x_m$ preserves $R$ as it can be obtained by composition from $x\cap y$ by $x_1\cap(x_2\cap (x_3\cap\cdots\cap(x_{m-1}\cap x_m )\cdots))$. In a bit more complicated fashion we can show that $x_0\cap(x_1\cup\cdots\cup x_m)$ preserves $R$
for every $m\geq 3$.
If $m=3$ it follows that $x_0\cap((x_0\cap(x_1\cup x_2))\cup x_3)$ is equal to $x_0\cap(x_1\cup x_2\cup x_3)$ (recall that $\cup$ and $\cap$ are the set union and intersection respectively). The pattern generalizes easily to arbitrary values for $m$. Finally, one obtains $f_{h,n}$  by suitably composing
$x_0\cap(x_1\cup\cdots\cup x_{{n}\choose{h}})$ and $x_1 \cap\cdots\cap x_h$.

Let $R$ be a relation in $\Gamma$ of arity, say, $r$ and let $t_1,\dots,t_{n}$ be a list of (not necessarily distinct) tuples in $R$. By the pigeon-hole principle, there exists a tuple $t$ appearing at least $\lceil n/|A|^r\rceil$ times in $t_1,\dots,t_{n}$. It follows from the choice of $h$ and $t$, that for every $I\subseteq\{1,\dots,n\}$, with $|I|=h$, there exists $i\in I$ such that $t=t_i$. It then follows  that $f_{h,n}(t,t_1,\dots,t_{n})$, which necessarily belongs to $R$, is precisely $g_{h,n}(t_1,\dots,t_{n})$
\end{proof}

For every natural number $n\in\mathbb{N}$, consider the $n$-ary fractional operation $\fas_{n}$ with support
$\left\{g_{h,n} \mid \left(1-\frac{1}{|A|^K}\right)n<h\leq n \right\}$ that distributes uniformly among the operations of its support.

\newcommand{\num}[2]{|{#1}|_{#2}}

\begin{lemma}
\label{le:jellyLipschitz}
There exists some $c\geq 0$ such that $\fas_{n}$  is  $c$-Lipschitz for every $n\in\mathbb{N}$.
\end{lemma}
\begin{proof}
Let $\mathbf{a}=(a_1,\dots,a_{n}),\mathbf{b}=(b_1,\dots,b_{n})\in A^{n}$.  Recall that from distributivity we assume that every element $a\in A$ is a subset of some set that we call $S$. Note that, according to the definition of $g_{h,n}$, an element $j\in S$ belongs to $g_{h,n}(\mathbf{a})$ iff $\num{\mathbf{a}}{j}\geq h$ where $\num{\mathbf{a}}{j}$ is defined to be $|\{ 1\leq i\leq n \mid j\in a_i\}|$. Consequently, $g_{h,n}(\mathbf{a})\neq g_{h,n}(\mathbf{b})$ iff there exists some $j\in S$ such that $\num{\mathbf{a}}{j}< h\leq\num{\mathbf{b}}{j}$ or
$\num{\mathbf{b}}{j}< h\leq\num{\mathbf{a}}{j}$ . It follows that

\begin{align*}
\Prob_{g\sim\fas_{n}}\{g(\mathbf{a})\neq g(\mathbf{b})\} & =\frac{\left|\left\{ h \mid  \exists j (
\num{\mathbf{a}}{j}< h\leq\num{\mathbf{b}}{j} \vee \num{\mathbf{b}}{j}< h\leq\num{\mathbf{a}}{j})\right\}\right|}{n/|A|^K} \\
& \leq \sum_{j\in S} \frac{\left|\left\{ h \mid
\num{\mathbf{a}}{j}< h\leq\num{\mathbf{b}}{j} \vee \num{\mathbf{b}}{j}< h\leq\num{\mathbf{a}}{j}\right\}\right|}{n/|A|^K} \\
& =
\frac{1}{n/|A|^K}\sum_{j\in S} |\num{\mathbf{a}}{j}-\num{\mathbf{b}}{j}|
\leq {|A|^K}\cdot |S|\cdot\dist(\mathbf{a},\mathbf{b}).
\end{align*}

\end{proof}

With the help of the sequence $\fas_{n}$, we can prove Theorem~\ref{the:algorithm}, i.e. obtain a constant-factor approximation algorithm for $\MinCSP\Gamma$.
A different proof of this result was given in~\cite{Dalmau13:robust}.

\begin{proof} (of Theorem~\ref{the:algorithm}).
Let $I=(V,A,{\mathscr C})$ be any instance of $\MinCSP\Gamma$ and let
$p_v (v\in V)$, $p_C (C\in{\mathscr C})$ be an optimal solution of $\lpmin(I)$
with objective value $\blpopt(I)$. We can assume that there exists some $n\in\mathbb{N}$ such that all the probabilities in the solution are of the form $n'/n$ where $n'$ is a non-negative integer. Also we can assume that $\log(n)$ is polynomial in the size of instance $I$.

Consider an assignment $s$ for $I$ obtained in the following way: draw  $g_{h,n}$ according to $\fas_{n}$ (i.e. select $\left(1-\frac{1}{|A|^K}\right)n<h\leq n$ uniformly at random) and assign $s(v)=g_{h,n}(p'_v)$ where $p'_v$ is any tuple such that every $a\in A$ appears exactly $p_v(a)\cdot n$ times in $p'_v$. It can be shown (this is basically the proof of direction $(2\Rightarrow 1)$ of Theorem \ref{the:Lipschitz}) that there exists some $c'\geq 1$ such that expected value of assignment $s$ is $c'\cdot \blpopt(I)$. In particular, $c'$ can be taken to be $2Kc$ where $c$ is the Lipschitz constant of $\fas_{n}$.

We shall prove that there is a randomized polynomial-time algorithm that constructs $s$. Select $\left(1-\frac{1}{|A|^K}\right)n<h\leq n$ uniformily at random. Recall that we assume that every element $a\in A$ is a subset of some set that we call $S$. Hence, in order to compute $g_{h,n}(p'_v)$, it is only necessary to give an efficient procedure that decides, for every $j\in S$, whether $j\in g_{h,n}(p'_v)$.
Note that, according to the definition of $g_{h,n}$, $j\in g_{h,n}(p_v')$ iff the number, $|p_v'|_j$, of entries in tuple $p'_v$ that contain $j$ is at least $h$. This number can be easily computed from $p_v$ as $|p_v'|_j=n\cdot \sum_{\{a\in A\mid j\in a\}} p_v(a)$.
\end{proof}

We finish this subsection by introducing another constraint language $\Gamma$ such that $\MinCSP\Gamma$ admits a constant-factor approximation algorithm. The interest of this result is in the fact that it is the first known example of a constraint language where $\MinCSP\Gamma$ has a constant-factor approximation algorithm but is not invariant under totally symmetric polymorphisms of all arities (i.e. $\Gamma$ does not have the so-called width 1 property~\cite{Feder98:monotone}).
This constraint language has domain $A=\{-1,0,+1\}$ and contains relations $R_{+}=\{(a_1,a_2,a_3)\in A^3 \mid a_1+a_2+a_3\geq 1\}$ and $R_{-}=\{(a_1,a_2,a_3)\in A^3 \mid a_1+a_2+a_3\leq -1\}$. This is the example in~\cite{Kun16} that we mentioned after Theorem~\ref{LP-sym}. It is easy to show that this constraint language has
no totally symmetric polymorphism of arity 3.


However $\{R_{+},R_{-}\}$ have many symmetric polymorphisms. In particular, it is not difficult to see that, for all $h,n\in\mathbb{N}$ with $h<\lfloor n/3\rfloor$, operation
$$s_{h,n}(x_1,\dots,x_n)=\begin{cases}
1 & \text{if } h< \sum_i x_i  \\
0 & \text{if }-h\leq\sum_i x_i\leq h  \\
-1 & \text{if }\sum_i x_i <-h
\end{cases}$$
preserves $\Gamma$.
It is also easy to show that the $n$-ary fractional operation with support $\{s_{h,n} \mid h<\lfloor n/3\rfloor\}$ that distributes uniformly among the operations of its support is $3$-Lipschitz and that can be efficiently sampled. Consequently,
$\MinCSP{\{R_{+},R_{-}\}}$ has a constant-factor approximation algorithm.

\subsection{NP-hardness result}
\label{sec:NU}

\newcommand{\independent}{\operatorname{Max}\ \operatorname{IS}}

In this subsection we prove Theorem~\ref{the:NU}, i.e.
show that, modulo P$\neq$NP, if $\MinCSP\Gamma$ admits a constant-factor-approximation algorithm then $\Gamma$ must have a near-unanimity (NU) polymorphism (recall the definition of an NU operation from Section \ref{sec:algebra}). NU polymorphisms
have been well studied in universal algebra~\cite{Baker75:chinese-remainder} and have been applied in CSP~\cite{Barto12:NU,Barto17:polymorphisms,Feder98:monotone,Dalmau17:robust}.
For example, every relation invariant under an $n$-ary NU operation is uniquely determined by its
$(n-1)$-ary projections~\cite{Baker75:chinese-remainder}, and NU polymorphisms characterize CSPs of ``bounded strict width''~\cite{Feder98:monotone}.



We can assume (proved in Lemma 3.7 of \cite{Dalmau13:robust}) that $\Gamma$ contains all unary singleton relations $\{a\}$, $a\in A$. This implies that polymorphisms of $\Gamma$ are idempotent.
It can be easily derived from Theorem~\ref{the:Lipschitz} that, modulo UGC, $\Gamma$ must have a near-unanimity polymorphism of some (large enough) arity. Indeed, for any $n$-ary fractional operation $\fas_{n}$ with support on symmetric polymorphisms of $\Gamma$ and every pair $a,b\in A$, the mass of operations $g$ in the support of $\phi_{n}$ such that $g(b,a,\dots,a)\ne g(a,a,\dots,a)\ (=a)$
is at most $\frac{c}{n}$. Since $c$ is constant, if we choose $n$ large enough, some $g$ in the support of $\phi_{n}$ will satisfy the near-unanimity identity.



In this section we shall prove it assuming only P$\neq$NP. As an intermediate step, we consider  the variant of $\CSP\Gamma$ where some constraints in an instance can be designated as
{\em hard}, meaning that they must be satisfied in any feasible solution, while the other constraints are soft and can be falsified.
It makes sense to
investigate approximation algorithms for this mixed version of CSP (see, e.g.~\cite{Guruswami16:mixed}).
In particular, the value of a feasible assignment for a instance of mixed $\MinCSP\Gamma$ is defined to be the number (or total weight) of soft constraints it violates.
It is not difficult to see, and was mentioned in~\cite{Guruswami16:mixed}
that mixed $\MinCSP\Gamma$ has a constant-factor approximation algorithm if and only if the ordinary, not mixed,
$\MinCSP\Gamma$ has such an algorithm.

The proof of our NP-hardness result makes use of a result about hardness of approximation for the problem $\independent_k$ in which
the goal is to find a maximum independent set in a given $k$-uniform hypergraph. Recall that an independent set in a hypergraph is a subset of its vertices that does not include any of its hyperedges (entirely).
 For real numbers $0\leq \alpha,\beta\leq 1$, say that an algorithm {\em $(\alpha,\beta)$-distinguishes} $\independent_k$ if, given a $k$-uniform hypergraph $H=(V,E)$, it correctly decides between the following two cases:
\begin{enumerate}
\item the size of the largest independent set of $H$ is at least $\beta\cdot |V|$
\item the size of the largest independent set of $H$ is at most $\alpha\cdot |V|$.
\end{enumerate}

\noindent
Note that it does not matter what the algorithm does for a hypergraph falling into neither of these cases.

\begin{theorem}[\cite{Dinur05:new}]\label{the:hyper}
For any integer $k\geq 3$ and any real number $\epsilon>0$, it is $\NP$-hard to $(\epsilon,1-\frac{1}{k-1}-\epsilon)$-distinguish $\independent_k$.
\end{theorem}


The key in proof of Theorem~\ref{the:NU} is to show that, roughly, if $\Gamma$ has no NU polymorphisms then $\Gamma$ can simulate (pp-define, to be precise),
for every $k\geq 3$, a $k$-ary relation $R_k$ such that
$R_k\cap\{a,b\}^k=\{a,b\}^k\setminus\{(a,\dots,a)\}$ for some distinct $a,b\in A$.
This relation, used in hard constraints, can encode a $k$-uniform hypergraph, while soft unary constraints
using relation $\{a\}$ simulate a choice of an independent set. To make this precise we will need a few definitions.

We say that $R$ is pp-definable from $\Gamma$ if there exists a (primitive positive) formula
$$\phi(x_1,\dots,x_k)\equiv \exists y_1,\dots,y_l \ \psi(x_1,\dots,x_k,y_1,\dots,y_l)$$
where $\psi$ is a conjunction of atomic formulas with relations in $\Gamma$ and $\eqC$ such that for every $(a_1,\dots,a_k)\in A^k$
$$(a_1,\dots,a_k)\in R \text{ if and only if } \phi(a_1,\dots,a_k) \text{ holds}.$$
Note that in the definition of primitive positive formulas we are slightly abusing   notation by identifying a relation with its relation symbol.
It is shown in~\cite{Dalmau13:robust} that if $\Gamma$ contains $\eqC$ and $R$ is pp-definable from $\Gamma$ then the problems $\MinCSP\Gamma$ and $\MinCSP{\Gamma\cup\{R\}}$ simultaneously belong or do not belong to $\APX$.

An $n$-ary operation on $A$ is called a {\em weak near-unanimity (WNU)} operation if it is idempotent and satisfies the identities
$$f(y,x,\dots,x)=f(x,y,\dots,x)=\cdots=f(x,x,\dots,y).$$

\medskip

\begin{proof} (of Theorem~\ref{the:NU}) Assume, towards a contradiction, that $\Gamma$ falsifies the statement of the theorem.

The following lemma can be derived from a combination of several known results. We give a (more or less) direct proof for completeness.

\begin{lemma}
\label{le:disjunctive}
For every $k\geq 1$, there is a $k$-ary relation, $R$, pp-definable from
$\Gamma$, and $a,b\in A$ such that
$$R\cap\{a,b\}^k=\{a,b\}^k\setminus\{(a,\dots,a)\}$$
\end{lemma}
\begin{proof} It follows easily from \cite{Baker75:chinese-remainder} that if $\pol(\Gamma)$ does not contain any NU operation, then for every $n\geq 3$ there is a relation $T\subseteq A^n$ which is pp-definable from $\Gamma$ and a tuple $(a_1,\dots,a_n)\not\in T$ such that for every $1\leq i\leq n$
there exists $c_i\in A$ such that $(a_1,\dots,a_{i-1},c_i,a_{i+1},\dots,a_n)\in T$.
Setting $n\geq (k+2)|A|^2$ it follows from the pigeon-hole principle that there exists
$a,c\in A$ and  $I=\{i_1,\dots,i_{k+2}\}\subseteq \{1,\dots,n\}$ of size $k+2$ such that $a_i=a$ and $c_i=c$ for every $i\in I$.
Consider relation $S$ defined as
$$S=\{(x_{i_1},\dots,x_{i_{k+2}}) \mid (x_1,\dots,x_n)\in T, \forall{i\!\not\in\! I}\, (x_i=a_i)\}$$
Clearly, $S$ is pp-definable using $T$ and the unary singletons. It follows that $S$ is pp-definable from $\Gamma$ as well.
We have  that $(a,a,\dots,a)\not\in S, t_1=(c,a,\dots,a)\in S$, $t_2=(a,c,\dots,a)\in S$, $\dots$, and $t_{k+2}=(a,a,\dots,c)\in S$.
We can also assume that, in addition to the previous property, $S$ is symmetric, meaning that if $(x_1,\dots,x_{k+2})$ belongs to $S$ then so does any tuple obtained by permuting its entries. This is because we can always replace $S$ by the relation
$\{(x_1,\dots,x_{k+2}) \mid (x_{\sigma(1)},\dots,x_{\sigma(k+2)})\in S \text{ for every permutation } \sigma  \}$
which is pp-definable from $S$. Since, by assumption, we have that $\MinCSP\Gamma$ admits a constant-factor approximation algorithm it follows from Theorem~9 of ~\cite{Dalmau13:robust} that $\Gamma$ has a certain property, called {\em bounded width} (or else $\PT=\NP$). Theorem~2.8 in \cite{Kozik14:maltsev} states that this property implies that
$\pol(\Gamma)$ contains WNU polymorphisms $g_3,g_4$ of arity 3 and 4, respectively, such that $g_3(y,x,x)=g_4(y,x,x,x)$ holds for for every $x,y\in A$. The proof of Theorem~2.8 in \cite{Kozik14:maltsev} shows how to obtain $g_n$ for $n=3,4$, but the proof generalizes immediately to show that, for each $n\ge 3$,
$\Gamma$ has an $n$-ary WNU polymorphism $g_n$, of arity $n$, and the identity $g_n(y,x,\dots,x)=g_{n'}(y,x,\dots,x)$ holds for all $n,n'$.  


Let $b=g_n(c,a,\dots,a)$ and let $j$ be minimum with the property that $S$ contains every tuple $t\in \{a,b\}^{k+2}$ with at
least $j$ $b$'s. We claim that $1\leq j\leq 3$. The lower bound follows from the fact that $(a,\dots,a)\not\in S$. For the upper bound, it follows from the fact every $g_n$ is a WNU (and so idempotent), that every tuple
$t\in \{a,b\}^{k+2}$ with $j(\geq 3)$ $b$'s can be obtained by  applying $g_j$ component-wise to tuples $t_{i_1},\dots,t_{i_j}$ where
$i_1,\dots,i_j$ are the components in $t$ that contain a $b$.  Since $S$ is symmetric
then it does not contain any tuple in $\{a,b\}^{k+2}$ with exactly $j-1$ $b$'s.

Finally, consider relation $R$ defined as
$$R=\{(x_1,\dots,x_k) \mid (\underbrace{b,\dots,b}_{j-1},\underbrace{a,\dots,a}_{3-j},x_1\dots,x_k)\in S\}$$
As before we infer that $R$ is pp-definable from $\Gamma$. It follows from the definition that $R$, $a$ and $b$ satisfy the statement of the lemma.
\end{proof}

\begin{lemma}
\label{le:reduction}
For every  $k\geq 1$, there is a linear algorithm  that, for a given $k$-regular hypergraph $H=(V,E)$, returns an instance $I$ of mixed $\MinCSP\Gamma$ such that the value of optimal solution for $I$ is $1-m/|V|$ where $m$ is the size of the maximum independent set in $H$.
\end{lemma}
\begin{proof} Fix $k\geq 1$ and let  $R$ and $a,b$ be as in Lemma \ref{le:disjunctive}. Let $\exists y_1,\dots,y_l \ \psi(x_1,\dots,x_k,y_1,\dots,y_l)$ be a primitive positive formula defining $R$ from $\Gamma$. It is well known that $\psi$ can be seen as an instance $J$ of $\CSP\Gamma$. More precisely, define $J$ to be the instance
that has variables $x_1,\dots,x_k,y_1,\dots,y_l$ and contains for every atomic formula $S(v_1,\dots,v_r)$ in $\psi$, the constraint $((v_1,\dots,v_r),S)$.  It follows that for any assignment $s:\{x_1,\dots,x_k,y_1,\dots,y_l\}\rightarrow A$, $s$ is a solution of $J$ if and only if $\psi(s(x_1),\dots,s(x_k),s(y_1),\dots,s(y_l))$ holds.

Consider the algorithm that, given a $k$-regular hypergraph, $H=(V,E)$, constructs an instance $I$ of mixed $\MinCSP\Gamma$ as follows. The set of variables of $I$ contains, in addition to all nodes in $V$, some other fresh variables to be introduced later. Then, for every hyperedge $E=\{v_1,\dots,v_k\}$, add a copy of $J$ where the variables have been renamed so that $x_1=v_1,\dots,x_k=v_k$ and $y_1,\dots,y_n$ are different fresh variables (different for each hyperedge). All the constraints added so far are designated as hard. Finally, add for every $v\in V$ a soft constraint $(v,\{a\})$ requiring $v$ to take value $a$.

Note that as $k$ is fixed, this can be carried out in linear time. It follows from the construction of $I$ that for every independent set, $X$, of $H$ there is an assignment for $I$ satisfying all hard constraints that maps every node in $X$ to $a$ and every node in $V\setminus X$ to $b$. This assignment violates exactly $|V|-|X|$ soft constraints. Conversely, for every assignment $s$ in $I$, the set $X=\{v\in V \mid s(v)=a\}$ is an independent set of $H$.
\end{proof}

\medskip

We are finally in a position to obtain a contradiction. As discussed above, if $\MinCSP\Gamma$ admits a constant-factor approximation algorithm then so does its mixed variant. Let $\delta$ satisfy $0<\delta/2\leq 1- c\cdot\delta$ where $c$ is any constant larger
than the approximation factor for the mixed $\MinCSP\Gamma$ algorithm. 
We shall prove that, for every $k$, there is a polynomial time algorithm that $(1-c\cdot\delta,1-\delta)$-distinguishes $\independent_k$, obtaining a contradiction since this task is NP-hard, as follows by setting $\epsilon=\delta/2$ and $k=1+2/\delta$ in Theorem \ref{the:hyper}.
Let us prove our claim. In order to distinguish whether the size, $m$, of the maximal independent set of a $k$-regular
hypergraph $H=(V,E)$ is at most $(1-c\cdot\delta)|V|$ or at least $(1-\delta)|V|$ we do the following:
compute the instance $I$ of mixed $\MinCSP\Gamma$ using the linear algorithm
of Lemma \ref{le:reduction} and run the constant approximation algorithm for $\MinCSP\Gamma$ with instance $I$.
Then, we only need to compare the value of the assignment, $s$, returned by the approximation algorithm with $c\cdot\delta$ to safely distinguish between the two cases.
Indeed, if $m\leq(1-c\cdot\delta)|V|$ then the optimum, $\opt$, of instance $I$, has value at least $c\cdot\delta$ from
which it follows that the value of $s$ is necessarily at least $c\cdot\delta$. Otherwise, if $m\geq (1-\delta)|V|$
then the value of $\opt$ is at most $\delta$, from which it follows that the value of $s$ is less
than $c\cdot\delta$.
\end{proof}

\section{Conclusion}

We have reduced a classification of constant-factor approximable finite-valued CSPs to that for Min CSPs. Due to technical limitations, we
proved most of our results for constraint languages $\Gamma$ containing the equality relation $\eqC$. It is in open question whether
adding $\eqC$ can ever change constant-factor approximability of $\MinCSP\Gamma$.
We provided (Theorem~\ref{the:Lipschitz}) an algebraic characterisation of constraint languages $\Gamma$ such that ($\Gamma$ contains $\eqC$ and) the integrality gap of BLP for $\MinCSP\Gamma$ is finite. We conjecture that $\MinCSP\Gamma$ is constant-factor approximable for all such languages, even without the assumption on $\eqC$. One way to prove this could be to strengthen the algebraic characterisation so that it features only fractional operations that can be efficiently sampled from. We showed that this works, in particular, for all constraint languages $\Gamma$ for which $\MinCSP\Gamma$ was previously known to be constant-factor approximable.

On the hardness side, infinite integrality gap
is known~\cite{Ene13:local} to imply UG-hardness of constant-factor approximation for all constraint languages $\Gamma$ containing the equality relation $\eqC$. For a large subclass of such languages, we improved UG-hardness to NP-hardness. Proving this for all such languages is (probably) beyond current techniques, even for very special cases such as {\sc MinUnCut}, but a further extension of our subclass could be within reach.

It is an open question whether condition (2) in Theorem~\ref{the:Lipschitz} is decidable. We remark that the decidability question for
the related property of having symmetric polymorphisms of all arities (see Theorem~\ref{LP-sym}) is also open - see~\cite{Carvalho16:symmetric,Chen17:asking} for related results. However, another related property - of having so-called fractional symmetric polymorphisms of all arities - is decidable~\cite{Kolmogorov15:power}.

\section{Acknowledgments}
The authors would like to thank Per Austrin, Johan H\aa stad, and Venkatesan Guruswami for useful discussions. We would also like to thank the anonymous referees for useful comments.

\bibliographystyle{elsarticle-num}
\bibliography{csp2}


\end{document}